\documentclass[12pt]{article}
\usepackage{makeidx}
\usepackage{amsfonts}
\usepackage{amssymb}
\usepackage{amsmath}
\usepackage[numbers,sort&compress]{natbib}
\usepackage[colorlinks,citecolor=blue,urlcolor=blue]{hyperref}
\usepackage{graphicx,indentfirst,tabularx}

\newtheorem{theorem}{Theorem}

\newtheorem{algorithm}[theorem]{Algorithm}

\newtheorem{lemma}[theorem]{Lemma}

\newenvironment{proof}[1][Proof]{\noindent\textbf{#1.} }{\ \rule{0.5em}{0.5em}}
\input{tcilatex}
\begin{document}

\title{A Monte Carlo Estimator for an Isolated Polynomial Zero via Contour
Integral Representations}
\author{Athanasios C. Micheas \thanks{%
Department of Statistics, University of Missouri, 146 Middlebush Hall,
Columbia, MO 65211-6100, USA, email: micheasa@missouri.edu}}
\maketitle

\begin{abstract}
We introduce a Monte Carlo approach for estimating isolated polynomial zeros
through their contour integral representations. For a simple zero enclosed
by an isolating contour, we show that the zero can be expressed exactly as
the expectation of a complex-valued random variable obtained by uniformly
sampling the contour. This yields an unbiased estimator without perturbing
the polynomial coefficients. We establish its variance and derive finite
sample probabilistic error bounds based on concentration inequalities, with
the bounds explicitly reflecting the geometry of the contour and its
separation from the zeros. The framework also provides a stochastic
formulation of contour based root counting and root isolation. Numerical
results illustrate the estimator's behavior under varying polynomial
degrees, and we further provide comparisons to existing methods.
\end{abstract}



\section{Introduction}

Let $\mathbb{C}$ denote the complex plane, and define a polynomial of degree 
$N\in \mathbb{N}^{+}=\{1,2,...\},$ over $\mathbb{C}$ by%
\begin{equation}
P_{N}(z)=a_{0}+a_{1}z+a_{2}z^{2}+...+a_{N}z^{N}=\sum%
\limits_{n=0}^{N}a_{n}z^{n},  \label{Poly1}
\end{equation}%
where the coefficients are also complex $a_{0},a_{1},a_{2},...,a_{N}\in 
\mathbb{C}.$ Let $r_{1},$ $r_{2},...,$ $r_{N}$ denote the $N$ roots of $%
P_{N}(z),$ and assume for now that the roots are simple and nonzero, i.e.,%
\begin{equation}
0<\left\vert r_{n}\right\vert <R,  \label{rootsorder}
\end{equation}%
where $\left\vert \cdot \right\vert $ denotes modulus, for some (large
enough) $R>0$, so that we may write%
\begin{equation}
P_{N}(z)=a_{N}(z-r_{1})(z-r_{2})...(z-r_{N})=a_{N}\prod\limits_{n=1}^{N}%
\left( z-r_{n}\right) ,  \label{Poly2}
\end{equation}%
for all $z\in \overline{B}(0;R):=\left\{ y\in \mathbb{M}:\left\vert
y\right\vert \leq R\right\} \subset \mathbb{C},$ the closed disc centered at
the origin $0+i0$ of radius $R>0$. The standard choice here is Cauchy's root
bound given by%
\begin{equation}
R=1+\max_{0\leq n<N}\left\vert \frac{a_{n}}{a_{N}}\right\vert ,
\label{rootsR}
\end{equation}%
such that the disc $\overline{B}(0;R)$ contains all roots of $P_{N}$.

For degrees $N\leq 4$, the roots of the general polynomial can be
represented by radicals. For $N\geq 5$, however, by the Abel--Ruffini
theorem, there is no formula by radicals that solves every polynomial of
degree $N$. The absence of a radical formula does not mean that the roots
cannot be represented analytically or computed numerically. In particular,
when $r_{i}$ is a simple root, the implicit function theorem implies
(locally in the coefficient variables), that the function $F_{i}$ defined by%
\begin{equation*}
r_{n}=F_{n}(a_{0},\ldots ,a_{N}),
\end{equation*}%
for $n=1,2,...,N$, is an analytic function. In particular, the central idea
is to combine the classical Cauchy contour representation of an isolated
simple root with randomized quadrature, yielding an unbiased Monte Carlo
estimator for an individual root. The construction is complemented by
probabilistic error bounds and a Monte Carlo formulation of the argument
principle for root counting. For some recent applications of Monte Carlo
methods see \cite{FoersterGrandePumiCybis2025}, \cite{FendRedenbach2026}, 
\cite{AlhelaliElDaliYusuf2026} and \cite{Mazo2026}.

To the best of our knowledge, the Cauchy contour representation of an
isolated polynomial zero has not previously been formulated explicitly as an
unbiased Monte Carlo estimator of the individual zero, with the randomness
introduced through contour sampling rather than through perturbation of the
polynomial coefficients. The proposed representation provides a direct
stochastic function of the polynomial coefficients%
\begin{equation}
r_{n}\approx \widehat{F}_{n}(a_{0},\ldots ,a_{N}),  \label{F_hat}
\end{equation}%
where $\widehat{F}_{n}$ is a Monte Carlo estimator of the root as a function
of the coefficients. In addition, we will obtain finite sample probabilistic
error bounds in terms of the geometry of the isolating contour. The
construction is complementary to existing polynomial root isolation and
randomized root finding algorithms; it does not replace the isolation
problem, but provides a statistically interpretable randomized approximation
once an isolating contour is available.

Specifically, \cite{Pan2024} develops randomized root finding algorithms
with nearly optimal complexity for the zeros of a univariate polynomial. The
algorithms can work either from coefficients or from a black-box evaluation
oracle and achieve nearly optimal complexity. There is, however, a technical
distinction: the algorithms in \cite{Pan2024} are Las Vegas randomized
algorithms, rather than Monte Carlo estimators in the statistical sense. A
Las Vegas algorithm uses randomness but guarantees correctness (with
randomness affecting running time), whereas the Monte Carlo algorithm
permits a small probability of an incorrect answer.

In \cite{PanZhao2017}, the authors investigate polynomial root finding
through randomized matrix and polynomial iterations. On the other hand,
there is work that studies polynomials whose coefficients are random. For
example, \cite{Shmerling2012} develops a Monte Carlo simulation algorithm
for studying the distribution of the number of zeros of random polynomials
in small intervals. For a review of recent work, along with precise
asymptotics of the variance of the number of real roots for a large class of
random polynomials see \cite{DoNguen2025} and the references therein.

While contour integration, Cauchy's formula, the residue theorem and the
argument principle or randomized polynomial root finding are well studied
areas of mathematics, in this paper, we propose the explicit combination of
these ingredients in a probabilistic framework in order to:

\begin{itemize}
\item propose an alternative construction from the standard randomized root
finding algorithms and standard Monte Carlo procedures that propagate
uncertainty in random polynomial coefficients, i.e., the polynomial is kept
deterministic and randomness is neatly introduced at the contour integration
representation of an isolated polynomial root,

\item provide a Monte Carlo representation of an isolated polynomial root as
the expectation of a contour sampling random variable,

\item obtain explicit finite sample probabilistic bounds for root estimation
based on the contour geometry and the polynomial coefficients,

\item give a Monte Carlo formulation of the argument principle for
estimating the number of roots inside a contour,

\item obtain a statistical distribution for the number of roots and the root
value inside an isolating contour, and

\item propose a Monte Carlo based algorithm for a probabilistically
controlled root isolation.
\end{itemize}

The paper proceeds as follows; in the next section we obtain a form of any
simple root of a polynomial in terms of a contour integral and utilize Monte
Carlo to obtain an unbiased estimator of $r_{n}$. We further derive
Hoeffding bounds that connect the numerical approximation error to contour
geometry, along with a combined isolation and root finding algorithm. In
Section 3, we utilize the results to perform a simulation study and
comparison with other methods. Concluding remarks are given in the last
section.

\section{Root representation by a contour integral}

Suppose that $r_{n}$ is a simple root of $P_{N},$ i.e., we have%
\begin{equation}
P_{N}(r_{n})=0,\text{ and }P_{N}^{\prime }(r_{n})\neq 0.  \label{simple-root}
\end{equation}%
Let $\Gamma _{n}$ be a positively oriented simple closed contour containing $%
r_{n}$ in its interior and no other root of $P_{N}$. Using equation (\ref%
{Poly2}), we can write 
\begin{equation}
\frac{P_{N}^{\prime }(z)}{P_{N}(z)}=\sum_{n=1}^{N}\frac{1}{z-r_{n}},
\label{logderivative}
\end{equation}%
and therefore 
\begin{equation*}
z\frac{P_{N}^{\prime }(z)}{P_{N}(z)}=\sum_{n=1}^{N}\frac{z}{z-r_{n}}.
\end{equation*}%
From the residue theorem (see \cite{conway2012functions}, Theorem 2.2, pg
108), since $r_{n}$ is simple, the residue of this function at $z=r_{n}$ is $%
r_{n}$. Because $\Gamma _{n}$ contains only $r_{n}$ and no other root, the
residue theorem gives 
\begin{equation}
r_{n}=\frac{1}{2\pi i}\oint\limits_{\Gamma _{n}}\sum_{n=1}^{N}\frac{z}{%
z-r_{n}}\,dz=\frac{1}{2\pi i}\oint\limits_{\Gamma _{n}}z\frac{P_{N}^{\prime
}(z)}{P_{N}(z)}dz.  \label{root-integral}
\end{equation}%
The latter equation is an exact representation of the root in terms of the
polynomial coefficients.

To make the representation more explicit, consider the circular contours%
\begin{equation}
\Gamma _{n}=\left\{ z=c_{n}+\rho _{n}e^{i\theta }:0\leq \theta \leq 2\pi
\right\} ,  \label{circle}
\end{equation}%
where the circle is such that it contains $r_{n}$ and no other root, for
some $\rho _{n}>0$ and $c_{n}\in \mathbb{C}$, that remain to be determined.
Suppose that $P_{N}$ has no zero on $\Gamma _{n}$. The argument principle
(see \cite{conway2012functions}, \S 3, pg 119) gives 
\begin{equation}
N(c_{n},\rho _{n})=\frac{1}{2\pi i}\oint\limits_{\Gamma _{n}}\frac{%
P_{N}^{\prime }(z)}{P_{N}(z)}\,dz,  \label{argument-principle}
\end{equation}%
where $N(c_{n},\rho _{n})$ is the number of zeros of $P_{N}$ in the open
disc 
\begin{equation*}
B(c_{n};\rho _{n}):=\left\{ y\in \mathbb{M}:\left\vert y-c_{n}\right\vert
<\rho _{n}\right\} \subset \mathbb{C},
\end{equation*}%
counted with multiplicity.

Setting%
\begin{equation}
z_{n}(\theta )=c_{n}+\rho _{n}e^{i\theta },  \label{z_ntransform}
\end{equation}%
gives 
\begin{equation*}
dz_{n}(\theta )=i\rho _{n}e^{i\theta }d\theta ,
\end{equation*}%
and therefore, the counting number of equation (\ref{argument-principle})
becomes 
\begin{equation}
N(c_{n},\rho _{n})=\frac{\rho _{n}}{2\pi }\int\limits_{0}^{2\pi }\frac{%
P_{N}^{\prime }(c_{n}+\rho _{n}e^{i\theta })}{P_{N}(c_{n}+\rho
_{n}e^{i\theta })}e^{i\theta }d\theta .  \label{theta-count}
\end{equation}%
The latter is equal to one by construction of $\Gamma _{n}$, i.e., $%
N(c_{n},\rho _{n})=1,$ since $\Gamma _{n}$ contains only the simple root $%
r_{n}$ and no other.

Similarly, substitution of (\ref{z_ntransform}) into (\ref{root-integral})
yields 
\begin{equation}
r_{n}=\int\limits_{0}^{2\pi }g_{n}(\theta )\frac{1}{2\pi }d\theta ,
\label{integral-theta}
\end{equation}%
where%
\begin{equation}
g_{n}(\theta )=\rho _{n}e^{i\theta }z_{n}(\theta )\frac{P_{N}^{\prime
}(z_{n}(\theta ))}{P_{N}(z_{n}(\theta ))}.  \label{g_n}
\end{equation}

Define the uniform random variable $\Theta \sim \func{Unif}\left( 0,2\pi
\right) .$ Then we can rewrite (\ref{integral-theta}) as the expectation%
\begin{equation*}
\mathbb{E}[g_{n}(\Theta )]=\int\limits_{0}^{2\pi }g_{n}(\theta )\frac{1}{%
2\pi }d\theta .
\end{equation*}%
Consequently, the simple root $r_{n}$ can be represented by the exact
formula 
\begin{equation}
r_{n}=\mathbb{E}[g_{n}(\Theta )]=\rho _{n}\mathbb{E}\left[ e^{i\Theta
}z_{n}(\Theta )\frac{P_{N}^{\prime }(z_{n}(\Theta ))}{P_{N}(z_{n}(\Theta ))}%
\right] .  \label{expectation}
\end{equation}%
and thus each simple isolated root is the expectation of an explicitly
computable complex-valued random variable. Therefore, the polynomial remains
fixed while randomness is introduced only through sampling of the contour
parameter.

\subsection{The Monte Carlo estimator}

Now consider a random sample of size $M$ from a $\func{Unif}\left( 0,2\pi
\right) ,$ i.e.,%
\begin{equation}
\Theta _{1},\ldots ,\Theta _{M}\overset{iid}{\sim }\func{Unif}\left( 0,2\pi
\right) .  \label{UnifSample}
\end{equation}%
and define%
\begin{equation}
Y_{n,m}:=g_{n}(\Theta _{m}),  \label{Y_nm}
\end{equation}%
for $m=1,\ldots ,M.$ The Monte Carlo estimator of $r_{n}$ is immediately
obtained by%
\begin{equation}
\widehat{r}_{n}^{(M)}=\widehat{F}_{n}^{(M)}(a_{0},\ldots ,a_{N})=\frac{1}{M}%
\sum_{m=1}^{M}Y_{n,m},  \label{mc-estimator}
\end{equation}%
or equivalently each root is approximated by 
\begin{equation}
\widehat{r}_{n}^{(M)}=\frac{\rho _{n}}{M}\sum_{m=1}^{M}e^{i\Theta
_{m}}(c_{n}+\rho _{n}e^{i\Theta _{m}})\frac{P_{N}^{\prime }(c_{n}+\rho
_{n}e^{i\Theta _{m}})}{P_{N}(c_{n}+\rho _{n}e^{i\Theta _{m}})}.
\label{explicit-estimator}
\end{equation}%
Note that the Monte Carlo estimator of equation (\ref{theta-count}) of the
roots inside the disc $B(c_{n};\rho _{n})$ is given by%
\begin{equation}
\widehat{N}_{M}(c_{n},\rho _{n})=\frac{\rho _{n}}{M}\sum_{m=1}^{M}e^{i\Theta
_{m}}\frac{P_{N}^{\prime }(c_{n}+\rho _{n}e^{i\Theta _{m}})}{%
P_{N}(c_{n}+\rho _{n}e^{i\Theta _{m}})},  \label{MCcountroot}
\end{equation}%
which is not an integer in general, and depending on machine arithmetic, it
is not even a real number; the imaginary part is close to zero, and tends to
zero as $M\rightarrow \infty $. Therefore, for practical implementations we
will take%
\begin{equation}
\widehat{N}_{M}(c_{n},\rho _{n})\approx \func{Re}\left[ \widehat{N}%
_{M}(c_{n},\rho _{n})\right] =\frac{\rho _{n}}{M}\sum_{m=1}^{M}\cos \left(
\Theta _{m}\right) \func{Re}\left[ \frac{P_{N}^{\prime }(c_{n}+\rho
_{n}e^{i\Theta _{m}})}{P_{N}(c_{n}+\rho _{n}e^{i\Theta _{m}})}\right] .
\label{realN}
\end{equation}%
In order to obtain an integer value we consider several approaches (see
Algorithm \ref{IsolationAlgorithm} below) and we discuss one in detail here.
Consider the probabilities of each possible value of $N(c_{n},\rho
_{n})=0,1,2,...,$ as follows; denote the indicator function by $I(\cdot ),$
and note that $\widehat{N}_{M}(c_{n},\rho _{n})$ is discrete random
variable, which we can be sampled using multiple samples%
\begin{equation}
\Theta _{j,1},\ldots ,\Theta _{j,M}\overset{iid}{\sim }\func{Unif}\left(
0,2\pi \right) ,  \label{uniformsamples}
\end{equation}%
for $j=1,2,...,L,$ and large integer $L>0$. Specifically, using (\ref{realN}%
), we have a sample from the distribution of $N(c_{n},\rho _{n})$ given by%
\begin{equation}
\widehat{N}_{M}^{(1)}(c_{n},\rho _{n}),\widehat{N}_{M}^{(2)}(c_{n},\rho
_{n}),...,\widehat{N}_{M}^{(L)}(c_{n},\rho _{n}),  \label{SampleN}
\end{equation}%
where%
\begin{equation*}
\widehat{N}_{M}^{(j)}(c_{n},\rho _{n})=\frac{\rho _{n}}{M}\sum_{m=1}^{M}\cos
\left( \Theta _{j,m}\right) \func{Re}\left[ \frac{P_{N}^{\prime }(c_{n}+\rho
_{n}e^{i\Theta _{j,m}})}{P_{N}(c_{n}+\rho _{n}e^{i\Theta _{j,m}})}\right] ,
\end{equation*}%
for $j=1,2,...,L.$ Then, appealing to Monte Carlo integration once again, we
can approximate the probability distribution of $N(c_{n},\rho _{n})$ using 
\begin{equation*}
\mathbb{P}(N(c_{n},\rho _{n})=n)=\mathbb{E}(I(N(c_{n},\rho _{n})=n))\approx 
\frac{1}{L}\sum_{j=1}^{L}I\left( \widehat{N}_{M}^{(j)}(c_{n},\rho
_{n})=n\right) ,
\end{equation*}%
for all $n=0,1,...,$ and therefore, we have estimators for the important
probabilities for no points in the disc $B(c_{n};\rho _{n})$ given by%
\begin{equation}
p_{0}:=\mathbb{P}(N(c_{n},\rho _{n})=0)\approx \frac{1}{L}%
\sum_{j=1}^{L}I\left( \widehat{N}_{M}^{(j)}(c_{n},\rho _{n})=0\right) ,
\label{ProbN0}
\end{equation}%
one point in the disc via%
\begin{equation}
p_{1}:=\mathbb{P}(N(c_{n},\rho _{n})=1)\approx \frac{1}{L}%
\sum_{j=1}^{L}I\left( \widehat{N}_{M}^{(j)}(c_{n},\rho _{n})=1\right) ,
\label{ProbN1}
\end{equation}%
and 2 or more points in the disc given by%
\begin{equation}
\mathbb{P}(N(c_{n},\rho _{n})\geq 2)=1-\mathbb{P}(N(c_{n},\rho _{n})=0)-%
\mathbb{P}(N(c_{n},\rho _{n})=1)=1-p_{0}-p_{1}.  \label{ProbN2}
\end{equation}%
These three probabilities can be used in order to propose and perform a
novel isolation algorithm (see Algorithm \ref{IsolationAlgorithm} below).

The process above using the samples (\ref{uniformsamples}) can also be used
to obtain samples from the distribution of the estimator of a root, given by%
\begin{equation}
\widehat{r}_{n}^{(1,M)},\widehat{r}_{n}^{(2,M)},...,\widehat{r}_{n}^{(L,M)},
\label{Samplern}
\end{equation}%
where%
\begin{equation*}
\widehat{r}_{n}^{(j,M)}=\frac{\rho _{n}}{M}\sum_{m=1}^{M}e^{i\Theta
_{j,m}}(c_{n}+\rho _{n}e^{i\Theta _{j,m}})\frac{P_{N}^{\prime }(c_{n}+\rho
_{n}e^{i\Theta _{j,m}})}{P_{N}(c_{n}+\rho _{n}e^{i\Theta _{j,m}})},
\end{equation*}%
for $j=1,2,...,L.$ Importantly, this allows us to estimate the mean and
variance of the Monte Carlo estimator of the root $r_{n}$ using%
\begin{equation*}
r_{n}\approx \frac{1}{L}\sum_{j=1}^{L}\widehat{r}_{n}^{(j,M)},
\end{equation*}%
which is expected to yield the value obtained in equation (\ref%
{explicit-estimator}), with variance given by%
\begin{equation}
\mathbb{E}\left[ \left\vert \widehat{r}_{n}^{(M)}-r_{n}\right\vert ^{2}%
\right] \approx \frac{1}{L}\sum_{j=1}^{L}\left\vert \widehat{r}_{n}^{(j,M)}-%
\frac{1}{L}\sum_{j=1}^{L}\widehat{r}_{n}^{(j,M)}\right\vert ^{2}.
\label{VarMCroot}
\end{equation}%
for all $n=1,2,...,N$.

Note that all the Monte Carlo estimators presented above are unbiased.
Indeed, since the samples (\ref{UnifSample}) are independent and identically
distributed, we can write 
\begin{equation*}
\mathbb{E}\left[ \widehat{r}_{n}^{(M)}\right] =\frac{1}{M}\sum_{m=1}^{M}%
\mathbb{E}[Y_{n,m}]=r_{n},
\end{equation*}%
and therefore, under the assumptions above, the estimator $\widehat{r}%
_{n}^{(M)}$ is unbiased. Similarly, for the number of roots, the Monte Carlo
estimator satisfies%
\begin{equation*}
\mathbb{E}\left[ \widehat{N}_{M}(c_{n},\rho _{n})\right] =N(c_{n},\rho _{n}),
\end{equation*}%
i.e., it is also unbiased, and the same arguments can be made for the
unbiased estimators of the probabilities (\ref{ProbN0}) and (\ref{ProbN1}).
These formulations involve an important distinction from Monte Carlo schemes
based on random perturbations of the coefficients; here there is no
perturbation bias, i.e., the expectation above is exactly the desired root
(and similarly for the number of roots or the probabilities).

Several benefits are immediately apparent from the Monte Carlo estimators of
equations (\ref{explicit-estimator}) and (\ref{MCcountroot}); firstly, the
calculations are naturally parallelizable because the contour evaluations
are independent. Second, unlike a purely deterministic numerical
approximation, the method naturally produces a sampling distribution for the
computed root, which can be useful for statistical uncertainty
quantification, and third, the representation may be attractive when root
computation is embedded within a larger stochastic computation, e.g.,
Metropolis-Hastings steps within a larger Gibbs sampler scheme.

\subsection{Probabilistic Error Bounds}

Define 
\begin{equation}
B_{n}:=\sup_{0\leq \theta \leq 2\pi }\left\vert g_{n}(\theta )\right\vert
=\sup_{0\leq \theta \leq 2\pi }\left\vert \rho _{n}e^{i\theta }z_{n}(\theta )%
\frac{P_{N}^{\prime }(z_{n}(\theta ))}{P_{N}(z_{n}(\theta ))}\right\vert ,
\label{B_n}
\end{equation}%
such that%
\begin{equation}
\left\vert g_{n}(\theta )\right\vert \leq B_{n},  \label{Bi-assumption}
\end{equation}%
and note that the real and imaginary parts are also bounded by the same $%
B_{n},$ i.e.,%
\begin{equation*}
\left\vert \func{Re}g_{n}(\theta )\right\vert \leq B_{n},\text{ and }%
\left\vert \func{Im}g_{n}(\theta )\right\vert \leq B_{n},
\end{equation*}%
for all $\theta \in \lbrack 0,2\pi ]$ and$\ n=1,2,...,N$. The following
gives an explicit form of the bounds $B_{n}$\ in terms of the polynomial
coefficients and geometric information about the contour.

\begin{lemma}
\label{LemmaBoundsB_n}Let 
\begin{equation*}
R_{n}:=\left\vert c_{n}\right\vert +\rho _{n},
\end{equation*}%
and suppose that the contour $\Gamma _{n}$ is separated from all roots by a
positive distance $d_{n},$ i.e.,%
\begin{equation}
\min_{z\in \Gamma _{n}}\min_{1\leq j\leq N}\left\vert z-r_{j}\right\vert
\geq d_{n}>0.  \label{d_n}
\end{equation}%
Then the upper bounds $B_{n}$ satisfy%
\begin{equation}
B_{n}\leq \frac{\rho _{n}R_{n}}{|a_{N}|d_{n}^{N}}%
\sum_{k=1}^{N}k|a_{k}|R_{n}^{k-1}.  \label{explicit-B}
\end{equation}%
for all $n=1,2,...,N.$
\end{lemma}

\begin{proof}
First note that for $z\in \Gamma _{n}$, we have%
\begin{equation*}
\left\vert z\right\vert \leq R_{n},
\end{equation*}%
and because 
\begin{equation*}
P_{N}^{\prime }(z)=\sum_{k=1}^{N}ka_{k}z^{k-1},
\end{equation*}%
after an appeal to the triangle inequality, we can write%
\begin{equation}
\left\vert P_{N}^{\prime }(z)\right\vert \leq \sum_{k=1}^{N}k\left\vert
a_{k}\right\vert R_{n}^{k-1}.  \label{Pprime-bound}
\end{equation}%
Using equations (\ref{Poly2}) and (\ref{d_n}), it follows that 
\begin{equation}
\left\vert P_{N}(z)\right\vert \geq |a_{N}|d_{n}^{N}.  \label{P-bound}
\end{equation}%
Consequently, we have 
\begin{equation*}
\left\vert g_{n}(\theta )\right\vert =\rho _{n}\left\vert z_{n}(\theta
)\right\vert \frac{|P_{N}^{\prime }(z_{n}(\theta ))|}{|P_{N}(z_{n}(\theta ))|%
}\leq \frac{\rho _{n}R_{n}}{|a_{N}|d_{n}^{N}}%
\sum_{k=1}^{N}k|a_{k}|R_{n}^{k-1}.
\end{equation*}%
and therefore (\ref{explicit-B}) is established.
\end{proof}

The following theorem presents some insight on the properties of the
constructed Monte Carlo estimators of the roots.

\begin{theorem}
\label{Thmbounds}Let $r_{n}$ be a simple root of $P_{N}$, and suppose that $%
\Gamma _{n}$ is a circular contour containing $r_{n}$ and no other root.
Suppose that (\ref{Bi-assumption}) holds with $B_{n}$ given by (\ref%
{explicit-B}), for each $n=1,2,...,N$, and for all $\theta \in \lbrack
0,2\pi ]$, and let $\widehat{r}_{n}^{(M)}$ be defined by (\ref%
{explicit-estimator}). Then for any $\varepsilon >0$ we have%
\begin{equation}
\mathbb{P}\left( \left\vert \widehat{r}_{n}^{(M)}-r_{n}\right\vert
>\varepsilon \right) \leq 4\exp \left( -\frac{M\varepsilon ^{2}}{4B_{n}^{2}}%
\right) .  \label{theorem-bound}
\end{equation}%
Consequently, if the sample size satisfies 
\begin{equation}
M\geq \frac{4B_{n}^{2}}{\varepsilon ^{2}}\log \frac{4}{\delta },
\label{M-bound}
\end{equation}%
for some $\delta >0,$ then 
\begin{equation}
\mathbb{P}\left( \left\vert \widehat{r}_{n}^{(M)}-r_{n}\right\vert
>\varepsilon \right) \leq \delta .  \label{epsilon-delta}
\end{equation}%
Under the assumptions above, if the sample sizes satisfy%
\begin{equation}
M_{n}\geq \frac{4B_{n}^{2}}{\varepsilon ^{2}}\log \frac{4N}{\delta },
\label{M_n-bound}
\end{equation}%
then we have the simultaneous bound 
\begin{equation}
\mathbb{P}\left( \max_{1\leq n\leq N}\left\vert \widehat{r}%
_{n}^{(M_{n})}-r_{n}\right\vert >\varepsilon \right) \leq \delta .
\label{all-roots}
\end{equation}
\end{theorem}

\begin{proof}
\textbf{Individual Bounds:} Using equation (\ref{Y_nm}), define 
\begin{equation*}
X_{n,m}:=\func{Re}Y_{n,m},
\end{equation*}%
and%
\begin{equation*}
V_{n,m}:=\func{Im}Y_{n,m}.
\end{equation*}%
Then equation (\ref{Bi-assumption}) yields 
\begin{equation*}
-B_{n}\leq X_{n,m}\leq B_{n},
\end{equation*}%
such that 
\begin{equation*}
\mathbb{E}[X_{n,m}]=\func{Re}r_{n},
\end{equation*}%
and similarly for the imaginary part we have%
\begin{equation*}
\mathbb{E}[V_{n,m}]=\func{Im}r_{n}.
\end{equation*}%
Hoeffding's inequality (see \cite{hoeffding1963probability}, \cite%
{hoeffding2012collected}), for the real part gives 
\begin{equation}
\mathbb{P}\left( \left\vert \func{Re}\widehat{r}_{n}^{(M)}-\func{Re}%
r_{n}\right\vert >t\right) \leq 2\exp \left( -\frac{Mt^{2}}{2B_{n}^{2}}%
\right) ,  \label{eq:hoeffding-real}
\end{equation}%
and similarly, for the imaginary part we obtain%
\begin{equation}
\mathbb{P}\left( \left\vert \func{Im}\widehat{r}_{n}^{(M)}-\func{Im}%
r_{n}\right\vert >t\right) \leq 2\exp \left( -\frac{Mt^{2}}{2B_{n}^{2}}%
\right) .  \label{eq:hoeffding-imag}
\end{equation}%
As a result, if 
\begin{equation*}
\left\vert \widehat{r}_{n}^{(M)}-r_{n}\right\vert >\varepsilon ,
\end{equation*}%
then either 
\begin{equation*}
\left\vert \func{Re}\left( \widehat{r}_{n}^{(M)}-r_{n}\right) \right\vert >%
\frac{\varepsilon }{\sqrt{2}},
\end{equation*}%
or 
\begin{equation*}
\left\vert \func{Im}\left( \widehat{r}_{n}^{(M)}-r_{n}\right) \right\vert >%
\frac{\varepsilon }{\sqrt{2}},
\end{equation*}%
and therefore, by the union bound, we have 
\begin{equation*}
\mathbb{P}\left( \left\vert \widehat{r}_{n}^{(M)}-r_{n}\right\vert
>\varepsilon \right) \leq 4\exp \left( -\frac{M\varepsilon ^{2}}{4B_{n}^{2}}%
\right) ,
\end{equation*}%
for all $n=1,2,...,N$, and (\ref{epsilon-delta}) is established.\newline
\textbf{Simultaneous Bounds: }For each $n$, choose 
\begin{equation}
M_{n}\geq \frac{4B_{n}^{2}}{\varepsilon ^{2}}\log \frac{4N}{\delta },
\label{M_n}
\end{equation}%
such that 
\begin{equation*}
\mathbb{P}\left( \left\vert \widehat{r}_{n}^{(M_{n})}-r_{n}\right\vert
>\varepsilon \right) \leq \frac{\delta }{N}.
\end{equation*}%
Using the union bound, we have 
\begin{equation*}
\mathbb{P}\left( \max_{1\leq n\leq N}\left\vert \widehat{r}%
_{n}^{(M_{n})}-r_{n}\right\vert >\varepsilon \right) \leq \sum_{n=1}^{N}%
\mathbb{P}\left( \left\vert \widehat{r}_{n}^{(M_{n})}-r_{n}\right\vert
>\varepsilon \right) \leq \delta ,
\end{equation*}%
and (\ref{all-roots}) is established, or in other words, with probability at
least $1-\delta $, the root estimators satisfy%
\begin{equation*}
\left\vert \widehat{r}_{n}^{(M_{n})}-r_{n}\right\vert \leq \varepsilon ,
\end{equation*}%
for all $n=1,\ldots ,N.$
\end{proof}

Substitution of (\ref{explicit-B}) into (\ref{M-bound}) gives the completely
explicit sample size requirement 
\begin{equation}
M\geq \frac{4}{\varepsilon ^{2}}\left[ \frac{\rho _{n}R_{n}}{|a_{N}|d_{n}^{N}%
}\sum_{k=1}^{N}k|a_{k}|R_{n}^{k-1}\right] ^{2}\log \frac{4}{\delta },
\label{explicit-M}
\end{equation}%
and thus, once an isolating contour and its separation from the roots are
known, the Monte Carlo sample size is explicitly controlled.

The Hoeffding bounds obtained in Theorem \ref{Thmbounds}, are convenient
because they require only boundedness of the function $g_{n}(\theta )$.
Alternatively, define the variance of $g_{n}(\Theta )$ by%
\begin{equation*}
\sigma _{n}^{2}:=\mathbb{E}\left[ \left\vert g_{n}(\Theta )-r_{n}\right\vert
^{2}\right] .
\end{equation*}

Since the samples are independent, the variance of the Monte Carlo estimator
(\ref{mc-estimator}) is given by%
\begin{equation*}
\mathbb{E}\left[ \left\vert \widehat{r}_{n}^{(M)}-r_{n}\right\vert ^{2}%
\right] =\frac{\sigma _{n}^{2}}{M},
\end{equation*}%
and consequently, a straightforward application of Markov's inequality
yields 
\begin{equation*}
P\left( \left\vert \widehat{r}_{n}^{(M)}-r_{n}\right\vert >\varepsilon
\right) \leq \frac{\sigma _{n}^{2}}{M\varepsilon ^{2}}.
\end{equation*}

Hence the sufficient condition 
\begin{equation}
M\geq \frac{\sigma _{n}^{2}}{\delta \varepsilon ^{2}},  \label{bound1}
\end{equation}%
also guarantees%
\begin{equation*}
\mathbb{P}\left( \left\vert \widehat{r}_{n}^{(M)}-r_{n}\right\vert
>\varepsilon \right) \leq \delta .
\end{equation*}%
In view of equation (\ref{VarMCroot}) we can obtain the Monte Carlo
estimator of $\sigma _{n}^{2}$ as%
\begin{equation}
\widehat{\sigma }_{n}^{2}=\frac{M}{L}\sum_{j=1}^{L}\left\vert \widehat{r}%
_{n}^{(j,M)}-\frac{1}{L}\sum_{j=1}^{L}\widehat{r}_{n}^{(j,M)}\right\vert
^{2}.  \label{MCvariance}
\end{equation}

In practice, Bernstein type inequalities or variance reduced Monte Carlo
methods can potentially improve the dependence on $\delta $ when the
variance is much smaller than the worst case bound $B_{n}^{2}$.

Furthermore, since the sample complexity depends on $B_{n}^{2}$, where the
bound for $B_{n}$ contains the term $d_{n}^{-N},$ when the contour
approaches a root or when roots become closely clustered, the theoretical
bound (\ref{M-bound}) can become extremely large. In particular, the
resulting sample size estimate%
\begin{equation*}
M=O\left( \frac{B_{n}^{2}}{\varepsilon ^{2}}\log \frac{1}{\delta }\right) ,
\end{equation*}%
should therefore be interpreted as a rigorous, worst case bound rather than
as a prediction of practical performance.

\subsection{Choice of circles and root isolation}

The results in the preceding sections assume that isolating contours are
available, and therefore the proposed root (and number of roots)
approximation depends on the choice of the circle centers $c_{n}\in \mathbb{C%
}$ and radii $\rho _{n}>0,$ for all $n=1,\ldots ,N.$ This assumption should
not be confused with the problem of finding the roots themselves; however,
if no isolating contours are known, an additional root isolation procedure
is required.

This is a central problem of numerical algebra where we need to determine
the locations of all the zeros and, in particular, isolate them, i.e., each
disc contains exactly one zero, or when roots are too close to separate
reliably in order to obtain a well separated cluster of zeros. The
fundamental idea here is to first determine how many zeros lie in a disc,
and then subdivide the disc until it contains exactly one zero, or until a
well separated root cluster has been identified. A general strategy is to
use the argument principle to count the number of roots inside candidate
regions.

This point of view is closely associated with the classic contour methods of 
\cite{DelvesLyness1967} and \cite{LynessDelves1967}, and has subsequently
been incorporated into modern subdivision and certified polynomial root
finding algorithms. There are different methods that use classic results
such as Pellet's theorem (see \cite{Pellet1881}), which is especially useful
for root isolation because it can provide information about both a disc and
a surrounding root free annulus.

A modification of the R function \textit{polyroot} can be used to quickly
assess if a disc contains a root (this R function implements the Jenkins and
Traub method, see \cite{JenkinsTraub1972Algorithm419}). A recent
implementation combining such methods in the literature, that provides
pseudo-code for all operations, can be found in \cite{Becker2018NearOptimal}
(in particular, see the $\mathbb{C}$ISOLATE algorithm). The root isolation
problem can be delegated to these certified algorithms based on subdivision,
Pellet-type tests, and argument principle methods (libraries exist written
in C++, e.g., MPSolve). In addition, libraries that return pairwise disjoint
complex intervals, each containing exactly one root (with multiplicities),
like SageMath, require strong assumptions on the coefficients, e.g., the
coefficient support is integer, or rational, or Gaussian-rational, or
algebraic, i.e., not arbitrary complex coefficients.

After an extensive search, the reader will be surprised to find out that
there is no R package to perform Pellet-type isolation for polynomial roots.
Therefore, we have implemented both Pellet's method for counting roots in
discs, as well as a novel probabilistic subdivision method, using R code.
The general isolation strategy using discs is summarized in the following
algorithm, where we used \textit{R code }and \textit{RCPP} for its
implementation, and the Monte Carlo estimators of equations (\ref%
{explicit-estimator}), (\ref{MCcountroot}), (\ref{ProbN0}) and (\ref{ProbN1}%
). The R and C++ code for this implementation along with all reproducible
simulations presented below, are included in the supplementary material. We
call the new method MCIsolate.

\begin{algorithm}[Root Isolation and Recovery]
\label{IsolationAlgorithm}Start with a disc 
\begin{equation*}
B_{0}=B(c_{0},\rho _{0})
\end{equation*}%
known to contain all roots. In view of equation (\ref{rootsR}), we choose $%
\rho _{0}=R.$ Obtain the uniform sample of equation (\ref{UnifSample}) for a
large $M$, and the samples (\ref{uniformsamples}) for large $L$. We do this
only once in the beginning of the isolation process.\newline
For every current disc $B(c,\rho )$ do the following:
\end{algorithm}

\begin{enumerate}
\item In order to find the number%
\begin{equation*}
K=N(c,\rho ),
\end{equation*}%
we obtain a $100(1-a)\%$ confidence interval for $K$ as follows; we take the
sample (\ref{SampleN}), calculate%
\begin{equation*}
K_{j}=\widehat{N}_{M}^{(j)}(c_{n},\rho _{n}),
\end{equation*}%
for $j=1,2,...,L,$ and order it in ascending order to obtain the percentiles
of the distribution of $\widehat{N}_{M}(c_{n},\rho _{n})$ given by%
\begin{equation*}
K_{(1)},K_{(2)},...,K_{(L)}.
\end{equation*}%
Then the interval%
\begin{equation}
\Delta _{L}=\left[ K_{\left( \left[ \frac{aL}{2}\right] \right) },K_{\left( %
\left[ \left( 1-\frac{a}{2}\right) L\right] \right) }\right] ,  \label{CI}
\end{equation}%
contains approximately $100(1-a)\%$ of the values of $K$. If we identify
more than one integers $\widehat{K}$ in $\Delta _{L},$ increase $\alpha $
(i.e., decrease the confidence) to make the length of the interval $\Delta
_{L}$ smaller, and repeat the process, until we obtain only a single $%
\widehat{K}$ in $\Delta _{L}.$ This approach guarantees a statistically
certified root counter, however, it is not as fast as other, less accurate
approaches.\newline
Alternatively, we can choose the value of the integer $K$ that corresponds
to the largest probability%
\begin{equation}
\mathbb{P}(K=0),\text{ }\mathbb{P}(K=1),\text{ and }\mathbb{P}(K\geq 2),
\label{Probs}
\end{equation}%
provided that there is a clear suggestion, e.g., $\mathbb{P}(K=1)\approx 1$,
then there is only one root.\newline
A fast version of the algorithm uses equation (\ref{realN}) where we take
the real part and round to the closest integer, such that%
\begin{equation}
\widehat{K}\approx \text{nearest integer to }\func{Re}\left[ \widehat{N}%
_{M}(c_{n},\rho _{n})\right] .  \label{Khat}
\end{equation}%
\newline
Now once $\widehat{K}$ has been obtained, we can perform the isolation test.

\item If $\widehat{K}=0$, discard the disc since it contains no roots inside
it.

\item If $\widehat{K}=1$, retain the disc as an isolating disc.

\item If $\widehat{K}\geq 2$, subdivide the disc into smaller regions and
repeat.\newline
The subdivision can be accomplished as follows. Suppose 
\begin{equation*}
B=B(c,r),
\end{equation*}%
contains several roots. We replace $B$ by smaller regions and apply the root
counting procedure (step 1) to those regions. The precise subdivision
geometry is not essential, however the important requirement is that every
possible root contained in the parent region must be contained in at least
one child region. One simple strategy is to cover the disc $B$ by several
smaller discs. For example, one can choose displaced centers such as 
\begin{equation*}
c+\frac{r}{2},\text{ }c-\frac{r}{2},\text{ }c+i\frac{r}{2},\text{ and }c-i%
\frac{r}{2}.
\end{equation*}%
However, care must be taken to ensure that the children cover the entire
parent region.\newline
A particularly convenient implementation is to use a square subdivision.
First enclose the current disc in a square, and then split that square into
four equal subsquares and associate a disc with each subsquare, for example
by taking its center as the disc center and half its diagonal as the disc
radius.\newline
More precisely, suppose the current region is the disc 
\begin{equation*}
\overline{B}(c,r)=\{z\in \mathbb{C}:\left\vert z-c\right\vert \leq r\},
\end{equation*}%
where $c\in \mathbb{C}$ and $r>0$. We first enclose $B(c,r)$ in the square 
\begin{equation*}
S(c,r)=[\func{Re}(c)-r,\func{Re}(c)+r]\times \lbrack \func{Im}(c)-r,\func{Im}%
(c)+r],
\end{equation*}%
which has side length $2r$. Dividing it into four equal subsquares produces
four squares of side length $r$. The centers of these four subsquares are 
\begin{equation*}
k_{1}=c+\frac{r}{2}+i\frac{r}{2},
\end{equation*}%
\begin{equation*}
k_{2}=c-\frac{r}{2}+i\frac{r}{2},
\end{equation*}%
\begin{equation*}
k_{3}=c+\frac{r}{2}-i\frac{r}{2},
\end{equation*}%
and 
\begin{equation*}
k_{4}=c-\frac{r}{2}-i\frac{r}{2}.
\end{equation*}%
Then the diagonal of each subsquare has length $r\sqrt{2}.$ Therefore, a
disc centered at the center of a subsquare and having radius equal to half
of its diagonal contains the entire subsquare. Hence the child radius is 
\begin{equation*}
r_{child}=\frac{r\sqrt{2}}{2}=\frac{r}{\sqrt{2}},
\end{equation*}%
and thus the four child discs are $B\left( k_{j},\frac{r}{\sqrt{2}}\right) ,$
for $j=1,2,3,4.$ Therefore, after $k$ subdivision levels, the radius is
given by%
\begin{equation*}
r_{k}=\frac{r}{(\sqrt{2})^{k}}=\frac{r}{2^{k/2}}.
\end{equation*}%
As a result, the regions become progressively smaller, allowing the
recursive algorithm to continue until the desired isolation or approximation
tolerance is achieved.\newline
Now since 
\begin{equation*}
B(c,r)\subseteq S(c,r),
\end{equation*}%
and the four subsquares form a partition of $S(c,r)$, every point of the
parent disc lies in at least one of the four subsquares. Furthermore, each
subsquare is contained in its corresponding child disc, and therefore we
have 
\begin{equation}
B(c,r)\subseteq \bigcup_{j=1}^{4}B\left( c_{j},\frac{r}{\sqrt{2}}\right) .
\label{balldivision}
\end{equation}%
Consequently, every root contained in the parent disc is contained in at
least one child disc. The child discs do not need to be disjoint; they may
also extend outside the original parent disc. Neither fact causes a problem
for the subdivision algorithm, since the essential requirement is that the
union of the children covers the parent region. We summarize the steps of
the isolation algorithm as follows:

\begin{enumerate}
\item Start with a parent disc $B(c,r)$.

\item Compute the number of roots contained in the disc.

\item If the root count is zero, discard the disc.

\item If the disc is sufficiently small and contains only one root, retain
it as a root isolating region, and use $r/2$ as its radius.

\item Otherwise, construct the four child discs and apply the root counting
procedure to each child disc.

\item Recursively subdivide every child that contains roots.\newline
Note that the subdivision process proposed here leads always to an isolating
disc for each root, unlike Pellet-type test methods that may be unresolved.
\end{enumerate}

\item Eventually we will obtain discs $B(c_{n},\rho _{n}),$ such that 
\begin{equation*}
B(c_{n},\rho _{n})\cap \{r_{1},\ldots ,r_{N}\}=\{r_{n}\},
\end{equation*}%
for each $n=1,\ldots ,N.$ These are isolating discs for the individual
roots, and can be used in equation (\ref{explicit-estimator}) to obtain the
Monte Carlo estimators $\widehat{r}_{n}^{(M)}$ of the roots $r_{n}$. As is
the usual approach in such algorithms, we also apply Newton's algorithm in
order to obtain a more refined approximation to the root based on the
initial estimator $\widehat{r}_{n}^{(M)}$ (using a maximum of up to 20
iterations and a tolerance of $10^{-14}$). A further refinement is applied
to the isolating disc radii by taking all of them to be half the minimum
distance between identified roots, i.e.,%
\begin{equation}
\rho _{n}=0.5\underset{k<j}{\min }\left\vert \widehat{r}_{k}^{(M)}-\widehat{r%
}_{j}^{(M)}\right\vert ,  \label{isoldist}
\end{equation}%
for all $n=1,\ldots ,N.$
\end{enumerate}

The advantage of this approach is that the number of roots or their values
need not be known in advance. Because for a contour $\Gamma $, the contour
integral%
\begin{equation*}
n_{\Gamma }:=\frac{1}{2\pi i}\oint\limits_{\Gamma }\frac{P_{N}^{\prime }(z)}{%
P_{N}(z)}\,dz
\end{equation*}%
equals the number of zeros of $P_{N}$ inside $\Gamma $, counted with
multiplicity (provided $P_{N}$ has no zero on $\Gamma )$, the aforementioned
recursive subdivision of the complex plane can be used to locate regions
containing exactly one root. Once such a region has been found, the Monte
Carlo estimator of equation (\ref{explicit-estimator}) can be applied.

\subsection{Roots with Multiplicity}

The restriction to simple roots is mathematically natural, however, we
discuss what happens for a root $r$ of multiplicity $m$. In that case, the
residue is%
\begin{equation*}
\func{Res}_{z=r}\frac{P_{N}^{\prime }(z)}{P_{N}(z)}=m,
\end{equation*}%
and hence, over an isolating contour $\Gamma $ that does not contain $r,$ we
can write%
\begin{equation*}
\frac{1}{2\pi i}\oint\limits_{\Gamma }\frac{P_{N}^{\prime }(z)}{P_{N}(z)}%
dz=m.
\end{equation*}

Similarly, the contour integral representation of the root is%
\begin{equation*}
\frac{1}{2\pi i}\oint\limits_{\Gamma }z\frac{P_{N}^{\prime }(z)}{P_{N}(z)}%
dz=mr,
\end{equation*}%
and therefore, the contour integral representation of the root assumes the
form%
\begin{equation}
r=\frac{\frac{1}{2\pi i}\oint\limits_{\Gamma }zP_{N}^{\prime }(z)/P_{N}(z)dz%
}{\frac{1}{2\pi i}\oint\limits_{\Gamma }P_{N}^{\prime }(z)/P_{N}(z)dz}=\frac{%
\oint\limits_{\Gamma }zP_{N}^{\prime }(z)/P_{N}(z)dz}{\oint\limits_{\Gamma
}P_{N}^{\prime }(z)/P_{N}(z)dz}.  \label{multiple-root}
\end{equation}

This suggests a natural extension to roots with multiplicity larger than 1,
however, the ratio of two Monte Carlo estimators is generally not unbiased.
Whilst the primary estimator is an unbiased estimator for an isolated simple
zero, multiple root estimation requires additional treatment and is outside
the main unbiased framework. Nevertheless, the numerator and denominator of (%
\ref{multiple-root}) are easily computed using their Monte Carlo estimators,
and a sample from the distribution of $r$ can be obtained. Then any type of
estimators and confidence intervals can be derived numerically as desired.
For this introductory paper, in the simulations that follow we will restrict
to simple roots.

\subsection{Method sensitivity and potential problems}

All isolation methods in the literature show similar sensitivities and
limitations (see \cite{Becker2018NearOptimal}, Section 1.2, for an excellent
discussion and references to related literature), and the same is true for
the proposed Monte Carlo isolation method MCIsolate. We summarize some of
these potential problems next, and how they can be addressed. The method
proposed is theoretically sensitive to: contours close to polynomial roots,
closely clustered roots, high polynomial degree, and poorly chosen or weakly
separated isolating contours.

In particular, the MCIsolate method heavily depends on having an isolating
contour already available for the cleanest root estimation result. The
proposed subdivision algorithm addresses this issue, but its finite sample
reliability depends on the accuracy of the Monte Carlo root counts; we
solved this issue via probabilistic approaches which is one of the novelties
of MCIsolate (see Algorithm \ref{IsolationAlgorithm}, equations (\ref{CI})
and (\ref{Probs})).

Other limitations include that the main unbiased estimator applies directly
only to simple roots; multiple root estimation uses a ratio of estimators
which is generally biased; the worst case concentration bounds may be
extremely conservative; numerical evaluation becomes unstable when $P_{N}(z)$
is small on the contour; and the contour points must avoid roots, since $%
P_{N}^{\prime }(z)/P_{N}(z)$ is undefined on a contour passing through a
root. The MCIsolate algorithm alleviates these issues by construction, and
borrows strength from the Monte Carlo framework, allowing us to obtain a
distribution of these important quantities, that can then be used to
quantify uncertainty in a probabilistic framework.

Finally, depending on the tolerance $\tau $\ used to identify zeros, e.g., a
point $z$ is a zero if $\left\vert P_{N}(z)\right\vert <\tau ,$ all such
algorithms may miss a root or add additional roots in the final set of
roots. It is a common approach to try and combine similar roots into one,
depending on the tolerance given, and we have implemented this post
processing approach as well (see the R and RCPP code in the supplementary
files), once we obtain a collection of valid roots.

\section{Simulation study and comparisons}

We illustrate the fast MCIsolate Algorithm \ref{IsolationAlgorithm} in
different scenarios, and we will use a tolerance of $\tau =10^{-14}$ in
order to identify a valid root. First we consider the easier case of a
polynomial of the 4rth degree with real coefficients (there are more methods
in the literature for this special case) given by%
\begin{equation*}
P_{4}(z)=z^{4}+6z+3,
\end{equation*}%
with coefficients%
\begin{equation*}
a_{0}=3,a_{1}=6,a_{2}=0,a_{3}=0,\text{ and }a_{4}=1,
\end{equation*}%
and 4 roots 
\begin{eqnarray*}
&&1.0579194-1.592819i,1.0579194+1.592819i, \\
&&-0.5113996+0.000000i,\text{ and }-1.6044392+0.000000i,
\end{eqnarray*}%
that are clearly separated. The Cauchy bound is $R=7$, so that all\ 4 roots
are in the disc $\overline{B}(0;7).$

We run MCIsolate and visualize the results in Figure \ref{ex1}. The left
plot illustrates the first subdivision level, where we can clearly see how
the discs of equation (\ref{balldivision}) are able to cover the the
original disc $\overline{B}(0;7).$ The middle plot presents the true roots
as red $\times $s and the MCIsolate proposed roots using squares, where we
can clearly see that the roots are recovered exactly. We further display the
isolating discs using the radius of equation (\ref{isoldist}), with $\rho
_{n}=0.5465198$, for all $n=1,2,3,4.$ The plot on the right presents the
distribution of the number of roots ($L=5000$), which is clearly Gaussian,
with a Shapiro-Wilk normality test statistic having value $W=0.99967$, and a 
$p$-value of $0.6081$.

\begin{figure}[t]
\centering \includegraphics[width=0.995\textwidth]{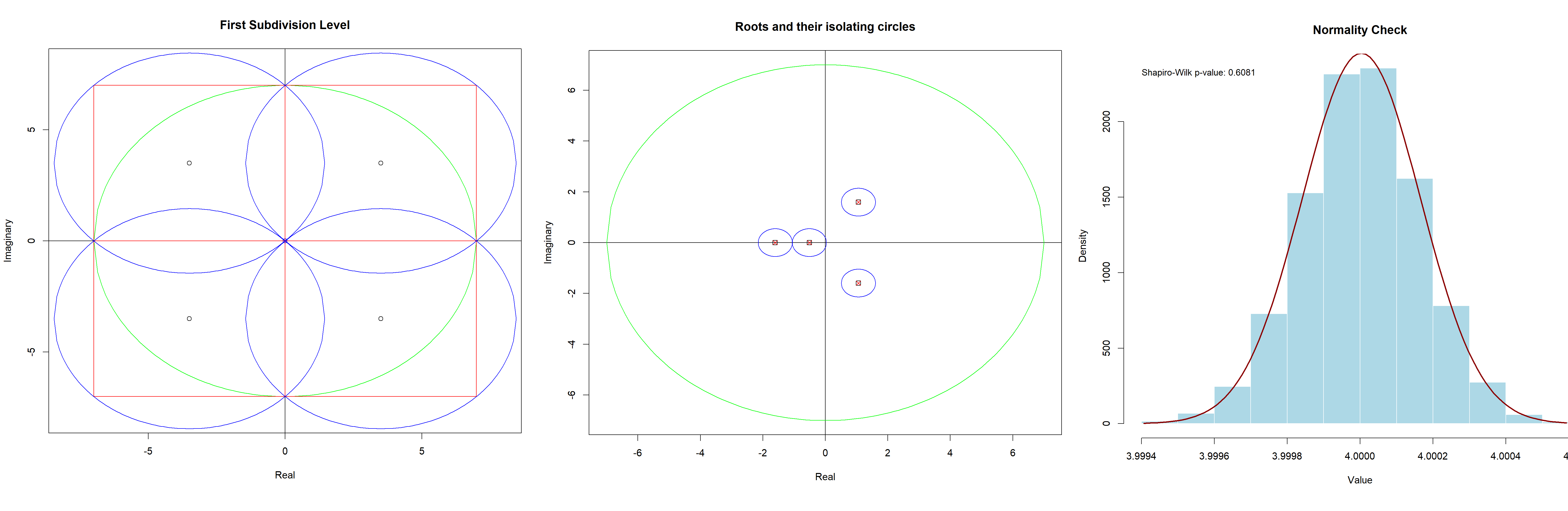}
\caption{Left: First subdivision level squares and discs. The Cauchy bound
disc is displayed as a green circle, the subdivision squares in red ink and
the subdivision discs in cyan ink. Middle: Displaying the true roots as red $%
\times $s and the MCIsolate proposed roots using squares. The isolating
discs are also displayed, with radii $\protect\rho _{n}=0.5465198$, for all $%
n=1,2,3,4.$ Right: Displaying the distribution of the number of roots, which
is clearly Gaussian.}
\label{ex1}
\end{figure}

\begin{figure}[th]
\centering \includegraphics[width=0.995\textwidth]{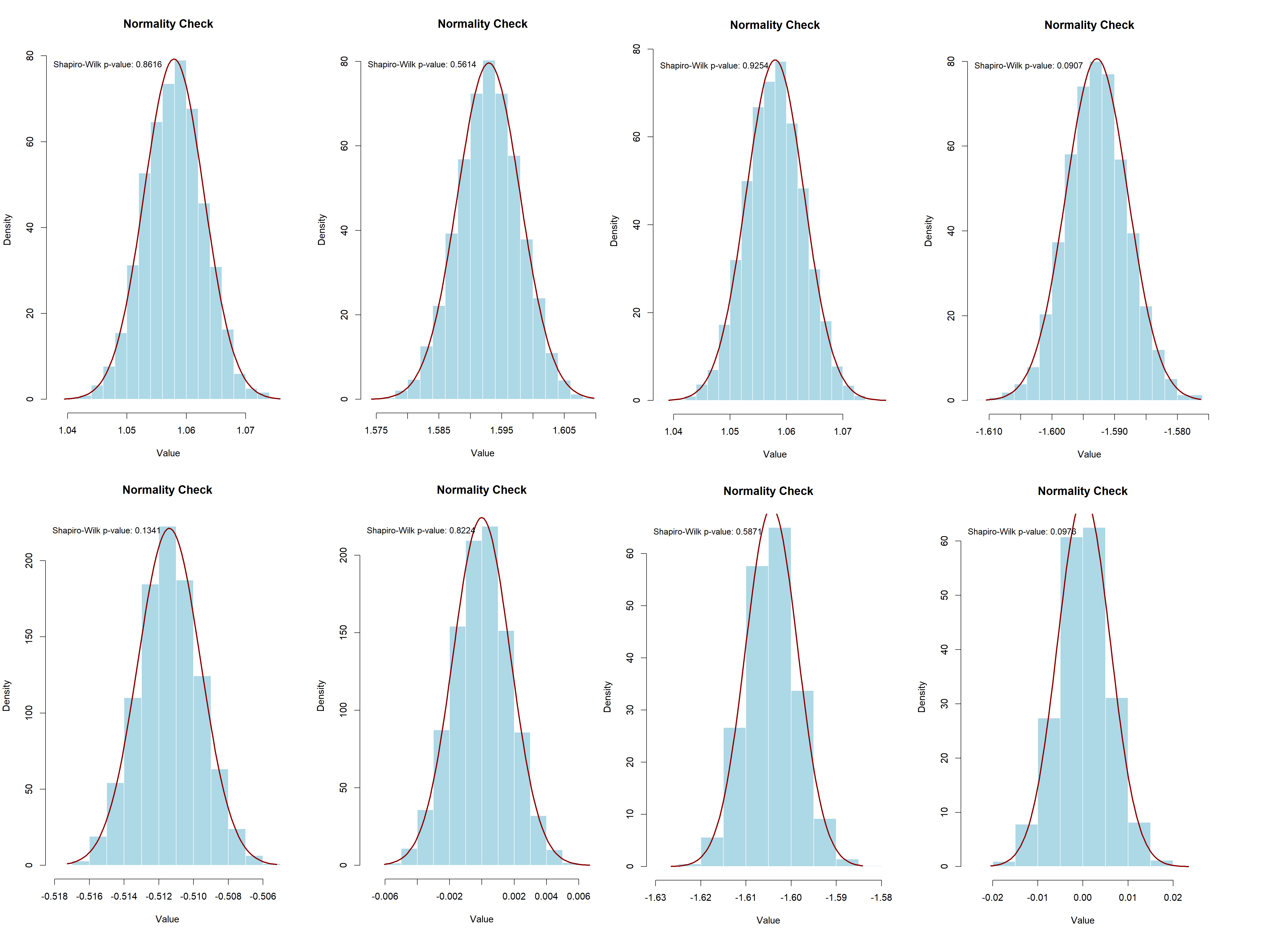}
\caption{Each pair of plots (from left to right) correspond to the
distributions of the real and imaginary parts of each of the four identified
roots. Notice that all root distributions are Gaussian.}
\label{ex1distrs}
\end{figure}

In Figure \ref{ex1distrs}, we present the distributions of the real and
imaginary parts of the four roots identified by the MCIsolate algorithm ($%
L=5000$). Note that in all cases, the distributions are Gaussian, with
approximate means the true values of the real and imaginary parts of the
roots. The Shapiro-Wilk test statistics and $p$-values are given by: for the
root $1.057924+1.592947i$ we have for the real part distribution $W_{\func{%
real}}=0.99976$, $p_{\func{real}}=0.8616,$ and for the imaginary part
distribution $W_{imag}=0.99965$, and $p_{imag}=0.5614$; for $%
1.057979-1.592849i$ we have $W_{\func{real}}=0.99979$, $p_{\func{real}%
}=0.9254$, and $W_{imag}=0.99939$, $p_{imag}=0.09071;$ for $%
-0.5114017-0.0000149i$ we have $W_{\func{real}}=0.99944$, $p_{\func{real}%
}=0.1341,$ and $W_{imag}=0.99974$, and $p_{imag}=0.8224$; and for $%
-1.604332+0.000214i$, we have $W_{\func{real}}=0.99966$, $p_{\func{real}%
}=0.5871,$ $W_{imag}=0.9994$, and $p_{imag}=0.09761.$

Next we compare MCIsolate with CIsolate of \cite{Becker2018NearOptimal},
which uses the Pellet/Rouch\'{e} test (see \cite{Pellet1881}, \cite%
{Marden1966}, \cite{Rump2003}, \cite{Melman2013}, \cite%
{BiniNoferiniSharify2013}). We use MCIsolate to obtain the Monte Carlo
estimator (using $M=50000,$ $L=5000$) for the root counts over the large
disc using Cauchy radius, for polynomials of degrees $5,10,15,...,100$, with
randomly selected complex coefficients, with $Unif(-10,10)$ independent
realizations for both the real and imaginary parts.

The results are summarized in Table \ref{Tablecomps}. Note that both methods
recover the true number of roots in every case, however, depending on the
tolerance value, it is possible that a method will report additional roots
that are near another root (3rd column). No post processing was applied for
both cases (e.g., combine nearby roots if within a certain threshold) and as
a result, we received the massive number of roots reported for the CIsolate
method. Even without post processing, the new method appears more robust to
this issue.

We also present the run times, where we note that as the polynomial degree
increases the MCIsolate times increase slower than the CIsolate run times.
Finally, note that after extensive simulations, we found that both methods
suffer from the machine arithmetic, and can produce NAs and Inf values in
the R and RCPP code, specifically in the simultaneous Horner evaluation of $%
P_{N}^{\prime }(z)/P_{N}(z)$ near a root, which is not unusual.

\newcolumntype{b}{>{\hsize=.1\hsize}X} \newcolumntype{e}{>{\hsize=.12%
\hsize}X} \newcolumntype{s}{>{\hsize=.34\hsize}X} 
\begin{table}[tbp]
\centering{\ 
\begin{tabularx}{\textwidth}{|b|e|s|s|}
\hline
$Degree$ & Cauchy $R$&MCIsolate vs CIsolate \# of reported roots& MCIsolate vs CIsolate Run Time  \\ \hline
5 & 2.389113 & 5 vs 2772& 1.68 vs 0.11 seconds \\\hline
10 & 2.199693& 10 vs 15925 & 5.31 vs 1.83 seconds \\\hline
15 & 2.11269& 15 vs 45227& 10.78 vs 11.59 seconds \\\hline
20 & 2.526373& 21 vs 178993& 30.07 vs 97.59 seconds \\\hline
25 & 3.953261& 28 vs 871243& 126.8 vs 814.4 seconds \\\hline
30 & 2.044319& 36 vs 389523& 102.33 vs 504.78 seconds \\\hline
35 & 5.6011& 79 vs 4913613& 633.22 vs 10816.08 seconds \\\hline
40 & 2.532048& 64 vs 1537036& 250.48 vs 3799.02 seconds \\\hline
45 & 2.612729& 101 vs 2424015& 422.66 vs 8733.15 seconds \\\hline
50 & 2.261706& 115 vs 2413228& 410.2 vs 11248.86 seconds \\\hline
\end{tabularx}
}
\caption{Comparisons between the MCIsolate and the CIsolate algorithms for
polynomials with degrees $N=5,10,...,50$. See the text for details and
discussion.}
\label{Tablecomps}
\end{table}

The computational cost of Monte Carlo sampling can be significant when a
large $M$ is required, which may be a computational disadvantage when
compared with the time required by deterministic methods for conventional
polynomial root finding. However, the proposed Monte Carlo formulation
provides unique features that deterministic methods are not able to provide,
including an unbiased estimator for the root or number of roots, explicit
sampling variability, finite sample error bounds, straightforward
parallelization, a stochastic formulation of contour root counting, and a
natural framework for uncertainty quantification. These features give the
MCIsolate method a distinct identity, that is unique when compared with
existing methods in the literature.

\section{Concluding Remarks}

The construction herein should be viewed as a Monte Carlo contour
integration estimator for an isolated root, rather than as a replacement for
all existing randomized polynomial root isolation algorithms. Our main
contribution is that the CMIsolate algorithm provides an explicit
statistical Monte Carlo formulation of the contour integral representation
of an isolated root, and in addition, the resulting unbiased estimator.

In particular, the principal reported results are theoretical where we
presented an exact expectation representation for every isolated simple
root; an unbiased sample average estimator; an unbiased Monte Carlo
estimator and formulation of the argument principle for estimating the
number of roots inside a contour; explicit variance and Hoeffding error
bounds based on the contour geometry and the polynomial coefficients;
simultaneous confidence guarantees for all roots; a geometry dependent bound
on the estimator magnitude; a recursive Monte Carlo root isolation strategy,
and an extension formula for roots with multiplicity. In terms of our
computational framework, we provided an algorithm for potential combination
of root isolation followed by Monte Carlo contour approximation.

The methods presented here can be utilized for finding the zeros of any
function via piecewise polynomial approximation, where the polynomial roots
approximate the roots of the original function. The function needs to
satisfy certain smoothness conditions in this case, e.g.,continuous. These
results are of great interest and will be presented elsewhere.

\section*{Declarations}

\begin{itemize}
\item I have no conflict of interest with the results in this paper.

\item This research received no funding.
\end{itemize}

\bibliographystyle{plainnat}
\bibliography{PolyRoots}

\end{document}